\documentclass[onecolumn,11pt]{article}
\usepackage{amssymb,amsmath,amsthm,bm}
\usepackage{graphicx}%
\usepackage{color}
\usepackage{fullpage}
\usepackage{threeparttable}
\usepackage{verbatim}

\theoremstyle{definition}
\newtheorem{definition}{Definition}[]
\theoremstyle{definition}
\newtheorem{theorem}[definition]{Theorem}
\newtheorem{lemma}[definition]{Lemma}

\theoremstyle{remark}

\usepackage{dcolumn}
\newcolumntype{d}{D{.}{.}{-1}}

\usepackage[symbol]{footmisc}

\usepackage{authblk} 

\newcommand{\etal}{\textit{et~al.}}

\newcommand{\afs}[1]{\alpha_\text{fs}^{#1}}
\newcommand{\afsnosup}{\alpha_\text{fs}}

\newcommand{\cmplxi}{\mbox{$\mathrm{i}$}}
\newcommand{\Rvdw}[1]{R_\text{vdW}^{#1}}

\newcommand{\A}{{\!A}}
\newcommand{\eff}{\text{eff}}
\newcommand{\free}{\text{free}}
\newcommand{\NR}{\text{NR}}
\newcommand{\R}{\text{R}}
\newcommand{\RAIM}{\text{RAIM}}
\newcommand{\TS}{\text{TS}}

\newcommand{\dirac}{\textsc{dirac}} 

\newcommand{\mc}[3]{\multicolumn{#1}{#2}{#3}}

\usepackage{titlesec}

\usepackage[sort&compress,numbers,comma]{natbib}

\usepackage{hyperref}
\hypersetup{colorlinks,breaklinks,linkcolor=Blue,urlcolor=Blue,anchorcolor=Blue,citecolor=Blue}
\usepackage{color}
\definecolor{DarkBlue}{rgb}{0.0,0.08,0.45}
\definecolor{Blue}{rgb}{0.0,0.0,1.0}
\definecolor{Red}{rgb}{1.0,0.0,0.0}
\definecolor{RedOrange}{rgb}{0.9,0.0,0.2}
\definecolor{dgrn}{RGB}{0,150,0}
\definecolor{dgray}{gray}{0.3}

\makeatletter
\def\@fnsymbol#1{\ensuremath{\ifcase#1 \or 1 \or 2 \or \dagger \or \ast  \or \ddagger  
\or \mathsection\or \mathparagraph\or \|\or **\or \dagger\dagger
   \or \ddagger\ddagger \else\@ctrerr\fi}}
\makeatother

\begin{document}
\title{
	Relativistic Hirshfeld atoms in a molecule: An information-theoretic view, with application to Drude oscillator dispersion models 
}
\author{
	Keegan Paice\footnote{Present address: Dept.\ of Chemistry \& Chemical Biology, McMaster University, Hamilton, Ontario, Canada}
	~and
	John M. Herbert\footnote{Corresponding author ({\tt herbert@chemistry.ohio-state.edu})}
}
\affil{
\em     Department of Chemistry \& Biochemistry \\
\em The Ohio State University, Columbus, Ohio USA
}
\date{}\maketitle

\begin{abstract}\noindent 
Several \textit{ad hoc} dispersion models for density-functional theory are based on the use of 
Hirshfeld (or ``stockholder'') partition of a molecular charge density, which provides an \textit{in situ} definition of atomic size.  
We show that a recently introduced ``optimized'' quantum Drude oscillator model for dispersion admits
a closed-form solution in terms of the Lambert $W$ function, whose branches
identify the compact and diffuse oscillator solutions. The compact solution
determines the $C_8$ (dipole--quadrupole) dispersion coefficient 
analytically from the free-atom polarizability, $C_6$ coefficient, and van der Waals radius, 
without any reference $C_8$ data. 
Next, we provide a formal basis for a relativistic version of the atoms-in-molecule Hirshfeld partition.
Using four-component Dirac-Hartree-Fock densities for isolated atoms 
defines a strictly positive deformation field
that carries the relativistic changes in atomic density into the Hirshfeld partition. 
A uniqueness theorem for the non-relativistic case is extended to relativistic Hirshfeld atoms and admits 
an asymptotic expansion through
quadratic order in the fine-structure constant. Normalization requires the relativistic 
density correction to reshape the reference atom while preserving its population.
Finally, four-component polarizabilities and $C_6$ coefficients are
reported for closed-shell atoms and ions, which 
supply the reference data required to extend atoms-in-molecules dispersion models into the heavy-element regime.
Periodic trends are observable in a scalar contraction factor that measures relativistic effects. 
\\ \mbox{ } 
\end{abstract}

\section{Introduction}
Hirshfeld~\cite{Hir77a} defined an atom in a molecule as a ``stockholder'' share of the molecular electron density, with dividends (of the
charge density) apportioned according to how each atom contributes in a superposition of free-atom densities, which is 
sometimes called the \textit{promolecule density}. 
This partition was later
shown to be the unique minimizer of the Kullback-Leibler (KL) divergence between the partitioned and reference densities~\cite{Aye06,Aye07}.
Therefore, an atom in a molecule is an information-theoretic construct rather than a region of space.
It is a way of dividing up what belongs to whom, given some prior belief about what each atom looked like before the molecule formed.
For an intermolecular, atoms-in-molecules dispersion model~\cite{HerDiSTka17,GraHer24a}, 
the reference has a physical role because the response of the separated fragments determines the long-range interaction; 
correlation determines the dispersion interaction~\cite{KhaGorMul25}.

The present work extends the Hirshfeld definition of an atom in a molecule make it consistent with the Dirac equation, 
for use in atoms-in-molecules models of intermolecular dispersion that use atomic size compression (and the concomitant scaling 
of free-atom polarizabilities) to determine \textit{in situ} atomic $C_6$ and $C_8$ coefficients~\cite{HerDiSTka17,GraHer24a}.
Dispersion models of this type include the Tkatchenko-Scheffler (TS) pairwise model~\cite{TkaSch09},
the many-body dispersion (MBD) model of Tkatchenko and co-workers~\cite{TkaAmbDiS13,AmbReiDiS14,GogKhaVac23}, 
and the Becke-Johnson exchange-dipole model (XDM)~\cite{BecJoh05a,BecJoh05b,BecJoh07}. 
Each of these requires free-atom polarizabilities as reference data, which must be extended into the relativistic regime in
order to have a satisfactory model of dispersion interactions for heavy elements. That extension is treated here. 

Mathematical results are presented in Section~\ref{sec:math}. 
After introducing standard Hirshfeld theory and the uniqueness theorem that comes from minimizing the KL divergence in
Section~\ref{sec:math:Hirshfeld}, we 
derive a closed-form solution for a certain ``optimized'' quantum Drude oscillator model of polarizability~\cite{GogKhaVac23}, 
in Section~\ref{sec:math:Drude}.
The solution (Theorem~\ref{thm:LambertW}) identifies the compact branch that was 
previously selected numerically~\cite{GogKhaVac23}, for use in the MBD model of dispersion.
Next, we use the variational characterization of the stockholder partition (as the unique minimizer
of the KL divergence against a promolecular prior)~\cite{Aye06,Aye07,BulVanAye07,HeiAyeVer18}, 
presented here as Theorem~\ref{thm:KL-stockholder} in slightly more general form and with a self-contained proof, 
to specialize to the Dirac limit, in which the relativistic stockholder weights
are the unique minimizer of a certain deformation functional (Lemma~\ref{lem:RAIM-KL}).
Finally, we obtain the asymptotic expansion of the minimum divergence in
the fine-structure constant $\afs{}$, through order $\afs{2}$, under stated regularity
assumptions (Theorem~\ref{thm:expansion}).

Numerical results are presented in Section~\ref{sec:calcs}.  There, we consider the ratios $\bar{\Lambda}_A$ of
relativistic (Dirac-Hartree-Fock, DHF) and non-relativistic (Hartree-Fock, HF) atomic volumes for isolated atoms $A$ across 
the periodic table, along with some atomic ions.  That volume scaling is a crucial ingredient for atoms-in-molecules dispersion
models, as it determines how the free-atom polarizability is modified by the molecular environment.  As a precursor to constructing
such models, the present results for $\bar{\Lambda}_A$ demonstrate the extent to which relativity modifies the polarizability.
Interesting periodic trends are evident in the ratios $\bar{\Lambda}_A$ as a function of atomic number. 

\section{Mathematical Results}\label{sec:math}
We first introduce standard Hirshfeld theory (Section~\ref{sec:math:Hirshfeld}) before illustrating how it is used to construct
atoms-in-molecules dispersion models (Section~\ref{sec:math:Drude}).  The latter discussion concludes with a new self-contained form for
the parameters in the recently introduced ``optimized'' quantum Drude oscillator model of dispersion~\cite{GogKhaVac23}.  A previous
result~\cite{Aye06,Aye07} connecting the Hirshfeld atom to variational minimization of the KL divergence is extended to the relativistic
regime in Section~\ref{sec:math:KL}.

\subsection{Standard Hirshfeld theory}\label{sec:math:Hirshfeld}
\subsubsection{Definitions and notation}

Let $(X,\mathcal{B},dx) := (\mathbb{R}^3,\text{Borel},\text{Lebesgue})$.
Let the atom index set be $\mathcal{A}:=\{1,\dots,M\}$.
Fix atomic electron populations $N_A>0$ and define the total electron number in the molecule to be 
\begin{equation}
	N := \sum_{A\in\mathcal{A}} N_A \; .
\end{equation}
Let the molecular electron density be $\rho \in L^1(X)$ with $\rho(x)\ge 0$ a.e., and 
\begin{equation}
	\int_X \rho(x)\,dx = N \; .
\end{equation}
Define a set of non-relativistic (NR) \textit{shape functions} $s_A^{\mathrm{NR}} \subset L^1(X)$, 
for atoms $A\in\mathcal{A}$, which satisfy 
$s_A^{\mathrm{NR}}(x)\ge 0$ a.e.\ and 
\begin{equation}
	\int_X s_A^{\mathrm{NR}}(x)\,dx = 1 \; .
\end{equation}
Up to normalization, these shape functions are simply the non-relativistic free-atom densities, $\rho_{A}^{0,\text{NR}}$:
\begin{equation}
	s_A^{\text{NR}}(x) := \rho_{A}^{0,\text{NR}}(x) / N_A \; .
\end{equation}
The non-relativistic promolecule density is 
\begin{equation}\label{eq:NR-promolecule}
	\rho^{0,\text{NR}}(x) := \sum_{A\in\mathcal{A}} \rho_{A}^{0,\text{NR}}(x) \; .
\end{equation}
We assume that $\rho^{0,\text{NR}}(x)>0$ a.e. on $\{\rho>0\}$.  The non-relativistic Hirschfeld weight function for atom~$A$ is 
\begin{equation}\label{eq:wNR-normalization}
	w_A^\text{NR}(x) := \frac{\rho_{A}^{0,\mathrm{NR}}(x)}{\rho^{0,\text{NR}}(x)} \; .
\end{equation}
It is normalized such that 
\begin{equation}
	\sum_{A\in\mathcal{A}} w_A^{\mathrm{NR}}(x)=1 \ \text{a.e.} .
\end{equation}
Finally, the non-relativistic Hirshfeld partition is  
\begin{equation} 
	\rho_A^{\text{NR}}(x) := w_A^{\text{NR}}(x)\,\rho(x) \; .
\end{equation}

\subsubsection{Main result}

We next introduce a result that is originally due to Ayers and collaborators~\cite{Aye06,Aye07,BulVanAye07,HeiAyeVer18}, namely, 
a variational characterization of the stockholder partition as the unique minimizer of the KL divergence.  This result 
is restated here as Theorem~\ref{thm:KL-stockholder}, 
in slightly more general form (with an arbitrary nonnegative reference $q_A$ 
on a $\sigma$-finite measure space) and with a self-contained proof.

\begin{theorem}[KL stockholder minimizer; Ayers \textit{et al}.]\label{thm:KL-stockholder}
Let $(X,\mathcal B,dx)$ be a $\sigma$-finite measure space and let the index set be
$\mathcal{A}=\{1,\dots,M\}$.
Fix $\rho\in L^1(X)$ with $\rho\ge 0$ a.e. and define the feasible set
\begin{equation}
	\mathcal S := \Big\{ (\rho_A)_{A\in \mathcal{A}}\in \prod_{A\in \mathcal{A}}L^1(X)\,:\,
	\rho_A\ge 0\ \text{a.e. and } \sum_{A\in \mathcal{A}}\rho_A=\rho\ \text{a.e.}\Big\}.
\end{equation}
Let $q_A$ (with $A\in \mathcal{A}$) be measurable functions with $q_A\ge 0$ a.e.\ and
$q_A\in L^1(X)$. These will be the Hirshfeld atomic electron populations and we set $q:=\sum_{A\in \mathcal{A}}q_A$.
Define the extended-valued functional
\begin{equation}\label{eq:L}
	\mathcal L\big[(\rho_A)_{A\in\mathcal{A}}\big] :=
	\sum_{A\in \mathcal{A}}\int_X \rho_A(x)\log\left[\frac{\rho_A(x)}{q_A(x)}\right] dx 
\end{equation}
with the conventions
\begin{equation}
	0\log\Big(\frac{0}{c}\Big):=0\ \ (c\ge 0),\qquad
	t\log\Big(\frac{t}{0}\Big):=+\infty\ \ (t>0) \; .
\end{equation}
Assume that 
\begin{equation}
	q(x)>0\quad\text{for a.e. }x\in\{\rho>0\} 
\end{equation}
and also
\begin{equation} 
	\int_X \rho\log\!\Big(\frac{\rho}{q}\Big) dx< \infty \; .
\end{equation}
(The latter condition simply excludes the degenerate case in which $\mathcal L$ is
identically $+\infty$ on $\mathcal S$, where every feasible point would minimize it, trivially.)
Under these conditions, $\mathcal L$ has a unique minimizer over $\mathcal S$ (up to $dx$-null sets), given by
\begin{equation}\label{eq:rhoA-star}
	\rho_A^\star(x)=\rho(x)\,\frac{q_A(x)}{q(x)}\qquad\text{for $x\in X$ a.e.} \; .
\end{equation}
Here, we interpret $\rho(x)\,q_A(x)/q(x)=0$ on $\{\rho=0\}$. Moreover,
\begin{equation}
	\min_{\mathcal S}\mathcal L=\int_X \rho(x)\log\left[\frac{\rho(x)}{q(x)}\right] dx \; ,
\end{equation}
which establishes the Hirshfeld partition as the unique minimizer of $\mathcal{L}$.
\end{theorem}

\begin{proof}
Fix $x\in X$ such that $\rho(x)$ and all $q_A(x)$ are finite and $q(x)>0$ whenever $\rho(x)>0$.
Consider the pointwise problem and minimize the functional 
\begin{equation} 
	F_x\big[(t_A)_A\big] := \sum_{A\in \mathcal{A}} t_A\log\left[\frac{t_A}{q_A(x)}\right]
\end{equation}
over $t_A\ge 0$ with $\sum_A t_A=\rho(x)$.
If $\rho(x)=0$, then the only feasible choice is $t_A=0$ and $F_x=0$.

Assume $\rho(x)>0$.
If $q_A(x)=0$ for some $A$, then any feasible choice with $t_A>0$ affords 
$t_A\log(t_A/0)=+\infty$, hence any finite minimizer must satisfy $t_A=0$ on those indices and we may safely 
restrict $A$ to indices with $q_A(x)>0$.

On the set $\{t_A>0\}$, $u\mapsto u\log u$ is strictly convex, hence $F_x$ is strictly convex
on the simplex $\{t_A\ge 0:\sum_A t_A=\rho(x)\}$, restricted to the aforementioned indices.
Introduce a Lagrange multiplier $\lambda$ for the constraint
\begin{equation}
\frac{\partial}{\partial t_A}\Big[t_A\log\Big(\frac{t_A}{q_A(x)}\Big)-\lambda t_A\Big]=0
\quad\Rightarrow\quad
\log\Big(\frac{t_A}{q_A(x)}\Big)+1-\lambda=0 \; .
\end{equation}
This affords 
\begin{equation}
	t_A = e^{\lambda-1}\,q_A(x) \; .
\end{equation}
Summing over $A$ gives $\rho(x)=e^{\lambda-1}q(x)$. Therefore, 
\begin{equation}
	t_A=\rho(x)\frac{q_A(x)}{q(x)} \; .
\end{equation}
By strict convexity, this pointwise minimizer is unique.

For $\rho_A^\star(x)$ as given in Eq.~\eqref{eq:rhoA-star}, $(\rho_A^\star)_A\in\mathcal S$.
For any $(\rho_A)_A\in\mathcal S$, pointwise minimality implies
$F_x[(\rho_A(x))_A] \ge F_x[(\rho_A^\star(x))_A]$ for $x$ a.e..
The functions $x\mapsto F_x[(\rho_A(x))_A]$ are the integrands that define $\mathcal L$ in Eq.~\eqref{eq:L}.
These functions are bounded below by an integrable function: minimizing
$u\mapsto u\log(u/c)$ over $u\ge0$ gives $u\log(u/c)\ge-c/e$, 
so $F_x\ge-q(x)/e$, and $q\in L^1(X)$ by hypothesis. Therefore, both sides are
well-defined in $(-\infty,+\infty]$ and the pointwise inequality may be
integrated, affording $\mathcal L[(\rho_A)_A] \ge \mathcal L[(\rho_A^\star)_A]$.
If equality holds, then the common value is finite by the assumption that 
$\int\rho\log(\rho/q)\,dx<+\infty$.  Therefore, the a.e.\ pointwise inequality may 
integrate to equality only if $F_x[(\rho_A(x))_A] = F_x[(\rho_A^\star(x))_A]$ for $x$ a.e..
Strict convexity at every $x\in\{\rho>0\}\cap\{q>0\}$ then forces
$\rho_A=\rho_A^\star$ a.e., proving global minimality and uniqueness up to null sets.

Finally, substituting $\rho_A^\star=\rho\,q_A/q$ gives
\begin{equation}
	\sum_A \rho_A^\star\log\Big(\frac{\rho_A^\star}{q_A}\Big)
	=\sum_A \rho\frac{q_A}{q}\log\Big(\frac{\rho}{q}\Big)
	=\rho\log\Big(\frac{\rho}{q}\Big) \; .
\end{equation}
Integration affords $\min_{\mathcal S}\mathcal L=\int \rho\log(\rho/q)$.
\end{proof}

This result will be generalized to the relativistic case in Section~\ref{sec:math:KL}.  
Before doing so, we first introduce how Hirshfeld partition is used in quantum Drude oscillator dispersion models.

\subsection{Quantum Drude oscillators}\label{sec:math:Drude}
\subsubsection{Tkatchenko--Scheffler scaling}\label{sec:math:Drude:TS}

The standard TS van der Waals (TS-vdW) 
scheme rescales free-atom response properties by the ratio 
of the \textit{in situ} atomic volume ($V_{\!A}^\eff$ for atom~$A$) to its corresponding free-atom value ($V_{\!A}^\free$).  That ratio is 
\begin{equation}\label{eq:TS}
	\gamma_{\!A}^\TS 
	:= \frac{ V_{\!A}^\eff }{ V_{\!A}^\free }	
	=\frac{\langle r^3 \rangle^{}_{\!A}}{\langle r^3 \rangle_{\!A}^0}
 \; .
\end{equation}
This quantity is used to define 
an effective (atoms-in-molecules) polarizability for atom~$A$, by taking advantage of the close correspondence of
polarizability and volume~\cite{BruMurPol93} and setting 
\begin{equation}
	\alpha_A^\eff = \gamma_A^\TS\,\alpha_A^\free \; .
\end{equation}
Pairwise $C_6$ dispersion coefficients are proportional to the polarizability of either partner, so we can define a $C_6$ coefficient 
for atom~$A$ as 
\begin{equation}
	C_{6,AA}^\eff = \big(\gamma_A^\TS\big)^2\,C_{6,AA}^\free \; .
\end{equation}
This accounts exclusively for geometric (volume) deformation of neutral atoms embedded in a molecular environment.
Relativistic deformation is added below.

The MBD model maps atomic polarizabilities onto harmonic (Drude) oscillators~\cite{CipCraSok19} having
parameters $\{m,q,\omega\}$, then solves the coupled many-body problem for their ground state. 
In the original MBD model~\cite{HerDiSTka17}, 
two of these parameters are fixed by atomic reference data that include polarizabilities and $C_6$
coefficients~\cite{JonCraSok13,GouBuc16,SchNag19}.
The remaining degree of freedom is eliminated either by setting $q=1$~a.u.\ or by supplying an \textit{ab initio} $C_8$ coefficient
and using the relationship $C_8 = 5\hbar C_6/\mu\omega$~\cite{JonCraSok13,CarLaoLiu19}.
However, computing $C_8$ means computing quadrupole polarizabilities, which is challenging~\cite{WooDun94,LaoYanDiS21,GogKhaVac23}.
The need for $C_8$ reference data is eliminated within the 
``optimized'' quantum Drude oscillator (QDO) formalism~\cite{GogKhaVac23}, which is described next.

\subsubsection{Optimized QDO formalism}\label{sec:math:Drude:opt}
In the optimized QDO parametrization~\cite{GogKhaVac23}, the final Drude oscillator 
parameter is not fixed by a reference $C_8$ value. Instead, the product $\mu\omega$ is determined by imposing that the QDO 
representation reproduces the same atomic radius ($R_\text{vdW}$) as that obtained from 
a universal relation between $R_\text{vdW}$ and the static dipole polarizability 
$[\alpha \equiv \alpha^{(1)}(0)]$.  The latter is derived from force balance in homonuclear dimers and given by~\cite{FedSadSto18}
\begin{equation}\label{eq:oqdo_alpha1_RvdW_universal}
	\alpha^{(1)}(0;\Rvdw{}) = \frac{ (4\pi\varepsilon_0) \Rvdw{7} }{ \big(a_0 / \afs{1/3}\big)^4 } 
\end{equation}
where $\afs{} :=e^2/(4\pi\varepsilon_0\hbar c)$ is the fine-structure constant and $a_0$ the Bohr radius.  The notation
$\alpha^{(1)}(0;\Rvdw{})$ in Eq.~\eqref{eq:oqdo_alpha1_RvdW_universal} indicates that the polarizability is evaluated in the static
(zero-frequency) limit, with $\Rvdw{}$ as a parameter.

Within the QDO model, the static polarizability admits the closed form
\begin{equation}\label{eq:oqdo_alpha1_RvdW_qdo}
	\alpha^{(1)}(0;\mu\omega,\Rvdw{}) =
	2^{7}\,(4\pi\varepsilon_0)\,
	R_{\mathrm{vdW}}^{\,7}\,
	\left(\frac{\mu\omega}{3\hbar}\right)^{\!2}
	\exp\!\left(-\frac{2\mu\omega R_{\mathrm{vdW}}^{\,2}}{\hbar}\right) \;, 
\end{equation}
adding $\mu\omega$ as another parameter.   The expressions in Eqs.~\eqref{eq:oqdo_alpha1_RvdW_universal} and 
\eqref{eq:oqdo_alpha1_RvdW_qdo} must be equal, which affords
\begin{equation}\label{eq:oqdo_equate}
	2^{7}\left(\frac{\mu\omega}{3\hbar}\right)^{\!2}
	\exp\!\left(-\frac{2\mu\omega \Rvdw{2}}{\hbar}\right)
	= \frac{\afs{4/3}}{a_0^{4}} \; .
\end{equation}
The positive square-root solution for $\mu\omega$ is a transcendental equation 
\begin{equation}\label{eq:oqdo_x_aexpbx}
	x = a\,e^{b x} \; ,
\end{equation}
expressed in terms of two dimensionless parameters
\begin{equation}\label{eq:a,b}
	a :=\frac{3\,\alpha_{\mathrm{fs}}^{2/3}}{8\sqrt{2}} \; ,
\qquad 
	b := \frac{ \Rvdw{2} }{ a_0^2 }
\end{equation}
and a dimensionless variable 
\begin{equation}\label{eq:x}
	x:=\frac{\mu\omega a_0^{2}}{\hbar} = \frac{a_0^2}{2\sigma^2}  \; .
\end{equation}
The quantity $\sigma^2 =\hbar/(2\mu\omega)$ in Eq.~\eqref{eq:x} 
is the variance of the QDO ground state~\cite{KhaFedAmb25}.

Let $W_k(z)$ denote the $k$th branch of the Lambert-$W$ function~\cite{CorGonHar96},
defined implicitly by
\begin{equation}\label{eq:LambertW_def}
	W_k(z)\,e^{W_k(z)} = z \; .
\end{equation}
Equation~\eqref{eq:oqdo_x_aexpbx} has the form
\begin{equation}\label{eq:oqdo_W_form}
	(-bx)e^{-bx}=-ab \; .
\end{equation}
and can be written compactly as 
\begin{equation}\label{eq:oqdo_y_W}
	y=W_k(-ab) \; ,
\end{equation}
where $y = -bx$. The solution is 
\begin{equation} \label{eq:oqdo_x_mu_solutions}
	x_k = -\frac{1}{b}W_k(-ab) \; .
\end{equation}
In terms of the desired QDO parameter $\mu\omega$, this is
\begin{equation}
	(\mu\omega)_k = \frac{\hbar x_k}{ a_0^2 }
	=-\frac{\hbar W_k(-ab)}{\Rvdw{2}} \; .
\end{equation}
This establishes the following result.

\begin{theorem}[Lambert-$W$ solution to the optimized QDO constraint]\label{thm:LambertW}
Let $W_k(z)$ denote the $k$th branch of the Lambert-$W$ function~\cite{CorGonHar96}, satisfying 
Eq.~\eqref{eq:LambertW_def}. For $a,b>0$ with $-1/e< -ab <0$, Eq.~\eqref{eq:oqdo_x_aexpbx} admits exactly two real solutions:
\begin{equation}
\label{eq:LambertW_sols}
x_k=-\frac{1}{b}\,W_k(-ab) \qquad \text{for}~k\in\{0,-1\} \; .
\end{equation}
These correspond to the principal branch ($k=0$) and the lower branch ($k=-1$).
\end{theorem}

\begin{proof}
The reduction in Eqs.~\eqref{eq:oqdo_W_form}--\eqref{eq:oqdo_x_mu_solutions} puts the
constraint in the canonical form,
\begin{equation}
y\,e^{y}=-ab,\qquad y:=-bx
\end{equation}
so that $y=W_k(-ab)$ on any branch whose domain contains the point $-ab$.
It remains to determine which branches those are. On the real line, the map
$z\mapsto ze^z$ takes values in $[-1/e,\infty)$, and for each $z\in(-1/e,0)$ it
has exactly two real pre-images: one in $(-1,0)$, which is $W_0$, and one in
$(-\infty,-1)$, which is $W_{-1}$. (At the endpoint $z=-1/e$, the two solutions coalesce at
$-1$ and the solution is unique.)
By hypothesis, $-ab<0$ and $-ab>-1/e$, so $-ab$ lies in that interval. Both
real branches therefore contribute, and no others do. 
Undoing the substitution $y=-bx$ affords Eq.~\eqref{eq:LambertW_sols}.
\end{proof}

\subsubsection{Branch identification and selection}
\label{sec:math:Drude:branch}

Since $W_0(-ab)\in[-1,0)$ and $W_{-1}(-ab)\le-1$ for $-ab\in[-1/e,0)$, the two
solutions satisfy $x_0\le 1/b\le x_{-1}$. The principal branch gives the smaller
$\mu\omega$ (diffuse oscillator, large $\sigma$), while the lower branch affords the larger
$\mu\omega$ (compact oscillator, small $\sigma$). This provides an \textit{a~priori} identification of the
two solutions found numerically by G\'oger  \etal~\cite{GogKhaVac23}. 
Their physical ``solution~A'', whose spread $\sigma$ tracks atomic vdW radii 
across the periodic table, is the $W_{-1}$ branch; the 
rejected ``solution~B'' is the $W_0$ branch. Indeed, for $ab\ll1$ one has
$W_0(-ab)=-ab+O((ab)^2)$, hence $x_0=a(1+O(ab))$.  On the principal branch, 
$\mu\omega \to \hbar a/a_0^2$ in this limit, with weak dependence on
$R_{\mathrm{vdW}}$. The closed form identifies the compact solution preferred
on the basis of atomic-response comparisons and the small-overlap assumption
in the optimized QDO construction~\cite{GogKhaVac23}:
$x_{-1}=-W_{-1}(-ab)/b$. The branch
boundary $-ab=-1/e$, at which the two solutions coalesce, corresponds to the critical polarizability 
($\alpha/4\pi\varepsilon_0 \approx 431$~a.u.) that was deduced by G\'oger  \etal~\cite{GogKhaVac23}.  
Neutral-atom reference polarizabilities considered in that work 
lie below this value. For a prescribed radius, the
real-solution condition is $aR_{\mathrm{vdW}}^2/a_0^2\le 1/e$.

\subsubsection{\texorpdfstring{Implicit $\bm{C}_{\bm8}$ coefficient}{Implicit $C_8$ coefficient}}

Set $\mathcal{H}:=L^2(X)$ and fix 
a self-adjoint Hamiltonian $\hat H$ on $\mathcal{H}$ with eigenpairs $\{(E_n,\lvert n\rangle)\}_{n\ge 0}$, where $E_0=\inf\sigma(\hat H)$.
For $L\in\mathbb{N}$, define the (scalar, isotropic) multipole polarizabilities at imaginary frequencies $\cmplxi\omega$ 
(with $\omega \ge 0$) by 
\begin{equation}\label{eq:alphaL_def}
	\alpha^{(L)}(\cmplxi\omega) :=
	\sum_{n\neq 0}
	\frac{2(E_n-E_0)\,|\langle 0 \lvert \hat Q^{(L)} \rvert n\rangle|^2}{(E_n-E_0)^2+\omega^2} \; .
\end{equation}
Quantities $\hat{Q}^{(L)}$ are the multipole operators.  The first two are 
$\hat Q^{(1)}:=\hat\mu$ (dipole) and $\hat Q^{(2)}:=\hat\Theta$ (quadrupole).

We use atomic units in the following response formulas. A single (isotropic)
QDO corresponds to the Hamiltonian
\begin{equation}\label{eq:oqdo_H}
	\hat H_\text{QDO} = \frac{\hat p^2}{2m} +\frac{1}{2}m\omega_0^2 \hat r^2
\end{equation}
with parameters $\{m,\omega_0,q\}$.  The dipole polarizability of this QDO is 
\begin{equation} \label{eq:alpha1_oqdo}
	\alpha^{(1)}(\cmplxi\omega) = \frac{q^2/m}{\omega_0^2+\omega^2}\; .
\end{equation}
For $\omega=0$, this affords the identity $\alpha = q^2/m\omega_0^2$~\cite{JonCraSok13} that was introduced in 
Section~\ref{sec:math:Drude:TS}.
The corresponding $C_6$ coefficient satisfies the Casimir-Polder relation in the uncoupled approximation~\cite{CasPol48,Lon65,Her21a},
\begin{equation} \label{eq:C6_oqdo}
	C_6 =
	\frac{3\hbar}{\pi}\int_0^\infty \big[\alpha^{(1)}(\cmplxi u)\big]^2\,du
	=
	\frac{3}{4}\alpha^2\omega_0 \; .
\end{equation}
Hence, $\omega_0 = 4C_6/3\alpha^2$ and $q^2/m=\alpha\omega_0^2$.

For the same optimized QDO, the quadrupole operator $r^2\mbox{Y}_{2m}$ connects the ground
state to the two-quantum manifold at excitation energy $2\hbar\omega_0$, so the
quadrupole polarizability is a single Lorentzian with its pole at $2\omega_0$:
\begin{equation} \label{eq:alpha2_oqdo}
	\alpha^{(2)}(i\omega) =
	\alpha^{(2)}(0)\left[\frac{(2\omega_0)^2}{(2\omega_0)^2+\omega^2}\right] 
\end{equation}
with 
\begin{equation}
	\alpha^{(2)}(0)
	= \frac{3\hbar\alpha}{4m\omega_0}
	=
	\frac{3}{2}\,\sigma^2\alpha \; .
\end{equation}
For two isotropic centers, the dipole--quadrupole dispersion receives both ordered
channels, $(1_A,2_B)$ and $(2_A,1_B)$, so the Casimir-Polder integral for the homonuclear $C_8$ reads
\begin{align}\label{eq:C8_oqdo}
\begin{aligned}
	C_8 &=
	\frac{15\hbar}{2\pi}\int_0^\infty \big[\alpha^{(1)}(\cmplxi u) \; \alpha^{(2)}(\cmplxi u)
	+\alpha^{(2)}(\cmplxi u) \; \alpha^{(1)}(\cmplxi u)\big]\,du
\\
	&=
	5\hbar\omega_0\,\alpha^{(1)}(0)\,\alpha^{(2)}(0) \; .
\end{aligned}
\end{align}
This can be rewritten as $C_8 = (15/2)\hbar\omega_0 \alpha^2\sigma^2$.
In deriving Eq.~\eqref{eq:C8_oqdo}, each ordered channel has contributed
$(5/2)\hbar\omega_0\alpha^{(1)}(0) \; \alpha^{(2)}(0)$ via
\begin{equation}
	\int_0^\infty\!du\,[(u^2+\omega_0^2)(u^2+4\omega_0^2)]^{-1} = \frac{\pi}{12\omega_0^3} \; .
\end{equation}
This result is consistent with the parameterization by Martyna and co-workers~\cite{JonCraSok13}, namely, 
$\mu=5\hbar C_6/(\omega_0 C_8)$~\cite{GogKhaVac23}, since
$C_8/C_6=10\sigma^2=5\hbar/(\mu\omega_0)$. 

This result implied that a QDO-parametrized form with inputs $\{\alpha,C_6\}$, along with the optimized width parameter $\sigma$, 
the $C_8$ coefficient is determined by the relation between its dipole and quadrupole response.
On the physical branch identified in Section~\ref{sec:math:Drude:branch}, the optimized variance is 
$\sigma^2=\hbar/(2\mu\omega)_{-1} =-R_\text{vdW}^2/2W_{-1}(-ab)$ in atomic units. Thus, Eq.~\eqref{eq:C8_oqdo} affords 
\begin{equation}
	C_8 =
	-\frac{15\alpha^2\omega_0\,R_{\mathrm{vdW}}^{\,2}}{4\,W_{-1}(-ab)}
\end{equation}
and no reference data for $C_{8,AA}$ are required.  This is likely to be especially useful in the case of ions, 
for which reference $C_8$ data are largely unavailable.  The explicit dependence of 
$a$ on the fine structure constant [Eq.~\eqref{eq:a,b}] means the optimized QDO mapping
$\{\alpha,C_6,\Rvdw{}\} \mapsto \mu\omega$ inherits a non-trivial dependence on $\afs{}$, 
providing a secondary route by which relativistic corrections propagate into dispersion models.
The primary route will be relativistic free-atom reference data, to be determined in Section~\ref{sec:calcs}.

\subsection{Kullback-Leibler minimization}
\label{sec:math:KL}

To place the quantum Drude oscillator dispersion model on a relativistic basis, we need a 
relativistic extension of Hirshfeld's definition of an atom in a molecule.  Following some general formalism in
Section~\ref{sec:math:KL:formalism}, we proceed to generalize 
Theorem~\ref{thm:KL-stockholder} to a relativistic setting in Section~\ref{sec:math:KL:lemma}.  Finally, a perturbative
expansion (in $\afs{}$) for the relativistic deformation of the free-atom density 
is derived in Section~\ref{sec:math:KL:perturbative}.

\subsubsection{Relativistic Stockholder Partition}
\label{sec:math:KL:formalism}

Let us first fix a relativistic reference shape function $s_A^\R(x;\afs{})$. (In practice, this will be obtained 
from a DHF free-atom density, in Section~\ref{sec:calcs}.) Define a \textit{deformation field} $\Lambda_A$ for atom~$A$ 
as the ratio of relativistic and non-relativistic shape functions, 
\begin{equation}\label{eq:LambdaA}
	\Lambda_A(x;\alpha_{\mathrm{fs}}) := 
	\frac{ s^\R_A(x;\afs{}) }{ s_A^\NR(x) } \; .
\end{equation}
(The fine-structure constant appears as a parameter here, in order to motivate the perturbative expansion in 
Section~\ref{sec:math:KL:perturbative}.)
We assume that $s_A^\R$ is strictly positive and satisfies conditions
\begin{subequations}
\begin{align}\label{eq:LambdaA-normalize}
	\int_X s_A^\NR(x) \; \Lambda_A(x;\afs{}) \; dx &= 1 
\\ \label{eq:LambdaA-limit}
	\lim_{\afsnosup \to 0}\Lambda_A(x;\afs{})&=1 \ \text{a.e.}\; .
\end{align}
\end{subequations}
We also assume throughout that $\rho^{0,\NR}(x)>0$~a.e.. 
This ensures that the Hirshfeld weight function 
\begin{equation}
	w_A^\NR(x) := \frac{ \rho_A^{0,\NR}(x)}{ \rho^{0,\NR}(x) }
\end{equation}
is well-defined a.e.\ on $\{\rho>0\}$. Since 
\begin{equation}
	\rho_A^{0,\R}=N_A s_A^\R = N_A s_A^\NR \Lambda_A=\rho_A^{0,\NR} \Lambda_A \; ,
\end{equation}
one obtains 
\begin{equation}
	\rho^{0,\R} = \sum_A \rho_A^{0,\R} = \rho^{0,\NR}\sum_A w_A^\NR \Lambda_A \; .
\end{equation}

Because $s_A^\text{NR}$ may vanish on a null set (or in practice due to truncation), we interpret the deformation field on the essential support
\begin{equation}
	\Omega_A := \{x\in X : s_A^\text{NR}(x) > 0\} \; .
\end{equation}
We say that $\Lambda_A(\,\bm\cdot\,;\afs{})$ is an admissible deformation if
$\Lambda_A:\Omega_A\to(0,\infty)$ is measurable and
\begin{equation}
	s_A^\R(x;\afs{})=s_A^\NR(x)\,\Lambda_A(x;\afs{})
	\quad \text{for a.e. } x\in X\; .
\end{equation}
Equivalently, the ratio $s_A^\R/s_A^\NR$ is taken only on $\Omega_A$, and $\Lambda_A$ can be extended arbitrarily on 
$X\setminus\Omega_A$ because it has no effect on any expression involving $s_A^\NR\Lambda_A$.

Alternatively, instead of fixing the relativistic reference shape $s_A^\R$, 
one could construct reference shape functions by defining relativistic reference densities as
\begin{subequations}
\begin{align}
	s_A^{\mathrm{R}}(x;\alpha_{\mathrm{fs}}) &:= s_A^{\mathrm{NR}}(x) \, \Lambda_A(x;\alpha_{\mathrm{fs}}) 
\\
	\rho_A^{0,\mathrm{R}}(x;\alpha_{\mathrm{fs}}) &:= N_A \, s_A^{\mathrm{R}}(x;\alpha_{\mathrm{fs}}) 
\end{align}
\end{subequations}
and the relativistic promolecule as
\begin{equation}
	\rho^{0,\mathrm{R}}(x;\alpha_{\mathrm{fs}}) := \sum_{A\in\mathcal{A}} \rho_A^{0,\mathrm{R}}(x;\alpha_{\mathrm{fs}}) \; .
\end{equation}
We prefer to maintain $\Lambda_A$ as the primitive object. Throughout, the molecular density $\rho$ is held fixed (as obtained from an
electronic structure calculation), and the relativistic atoms-in-molecules (RAIM) 
procedure modifies only the reference promolecule used to define stockholder weights.
Those are defined for the relativistic setting as 
\begin{equation}
	w_A^{\mathrm{R}}(x;\afs{}) := 
	\frac{\rho_A^{0,\text{R}}(x;\afs{})}{\rho^{0,\mathrm{R}}(x;\afs{})}
\end{equation}
with 
\begin{equation}
	\sum_A w_A^\text{R}(x;\alpha_{\mathrm{fs}})=1 \quad \text{a.e.}.  
\end{equation}	
The RAIM (or relativistic stockholder) partition is
\begin{equation}
	\rho_A^{\mathrm{R}}(x;\alpha_{\mathrm{fs}}) := w_A^{\mathrm{R}}(x;\alpha_{\mathrm{fs}})\,\rho(x).
\end{equation}

Let $\mathcal{S}$ denote the set of admissible densities,
\begin{equation}
	\mathcal{S} :=
	\left\{
	(\rho_A)_{A\in\mathcal A}\in \prod_{A\in\mathcal A} L^1(\mathbb{R}^3)
	\;:\;
	\rho_A(x)\ge 0 \ \text{a.e.},\ \
	\sum_{A\in\mathcal A}\rho_A(x)=\rho(x)\ \text{a.e.}
	\right\} \; .
\end{equation}
We can now define the relativistic functional
\begin{equation}
	\mathcal{L}_{\afs{}}\big[(\rho_A)_A\big] :=
	\sum_{A\in\mathcal{A}} \int_X \rho_A(x)\,
	\log\left[\frac{\rho_A(x)}{\rho_A^{0,\mathrm{R}}(x;\afs{})}\right] dx \; .
\end{equation}
This functional form follows directly from the information-theoretic formulation of the Hirshfeld 
partition~\cite{Aye06,BulVanAye07,HeiAyeVer18}.

\subsubsection{Uniqueness Result}
\label{sec:math:KL:lemma}

To apply Theorem~\ref{thm:KL-stockholder} in the relativistic setting,  we make the pointwise identifications
\begin{subequations}
\begin{align}
	q_A(x) &\equiv \rho_A^{0,\mathrm R}(x;\afs{})
\\
	q(x) &\equiv \sum_{B\in\mathcal A} q_B(x)=\rho^{0,\mathrm R}(x;\afs{}) \; .
\end{align}
\end{subequations}
We assume that $\int_X\rho\log(\rho/q)\,dx<\infty$ for the reference choices
under consideration, including the non-relativistic reference.
Under these identifications, the functional $\mathcal L_{\afs{}}$ is exactly of the form 
$\sum_A\int \rho_A \log(\rho_A/q_A)\,dx$, and Theorem~\ref{thm:KL-stockholder} yields the unique minimizer 
\begin{equation}
	\rho_A^\RAIM =\rho\,q_A/q=\rho\,\rho_A^{0,\text{R}}/\rho^{0,\text{R}}=\rho\,w_A^\R \; .
\end{equation}

The unique minimizer for RAIM over $\mathcal{S}$ is $\rho_A^\RAIM(x;\afs{})$, and we define its minimum value to be 
\begin{equation}
	\min_{\mathcal{S}}\mathcal{L}_{\afsnosup} =
	D(\alpha_{\mathrm{fs}}) :=
	\int_X \rho(x)  \log\left[\frac{\rho(x)}{\rho^{0,\mathrm{R}}(x;\afs{})}\right] dx \; .
\end{equation}
Note the identity
\begin{equation}\label{eq:auto:1}
	\rho^{0,\R}(x;\afs{}) =
	\rho^{0,\NR}(x)\sum_{A\in\mathcal{A}} w_A^\NR(x) \; \Lambda_A(x;\afs{}) \; ,
\end{equation}
which affords 
\begin{equation}\label{eq:auto:2}
	D(\alpha_{\mathrm{fs}}) = D(0) -
	\int_X \rho(x)\,
	\log\left[\sum_{A\in\mathcal{A}} w_A^\NR(x) \; \Lambda_A(x;\afs{})\right] dx
\end{equation}
where
\begin{equation} \label{eq:auto:3}
	D(0)=\int_X \rho(x)\,\log\left[\frac{\rho(x)}{\rho^{0,\NR}(x)}\right] dx \; .
\end{equation}

For dominated convergence, a sufficient condition is the existence of constants
$0<m\le M<\infty$ and $\alpha_\text{fs,max}^{}>0$ such that
\begin{equation}
\label{eq:auto:4}
	m \le \Lambda_A(x;\afs{})\le M \
	\text{a.e., for all }A\in\mathcal{A},\ |\afs{}|\le \alpha_\text{fs,max}^{} \; .
\end{equation}
That is, $\Lambda_A$ is bounded above and bounded away from zero uniformly in the
coupling ($\afs{}$) on all of $X$. A finite nuclear charge distribution regularizes the
behavior near the origin~\cite{VisDya97}; the whole-space bounds also require
control of the density ratios in the tails and are imposed here as hypotheses.

Under Eq.~\eqref{eq:auto:4}, the convex combination inherits the same bounds. The normalization condition in 
Eq.~\eqref{eq:wNR-normalization} implies that 
$m\le \sum_A w_A^{\mathrm{NR}}\Lambda_A \le M$. Therefore, 
\begin{equation}
	\left|\log\Big(\sum_A w_A^\NR \; \Lambda_A\Big)\right|\le C:=\max\big(|\log m|,|\log M|\big) \; .
\end{equation}
The integrand in Eq.~\eqref{eq:auto:2} is thus dominated by $C\rho$, which lies in $L^1(X)$
because $\int_X\rho\,dx=N<\infty$.

From the non-relativistic limit of $\Lambda_A$ in Eq.~\eqref{eq:LambdaA-limit}, 
the convex combination tends to unity a.e.\ and its
logarithm tends to zero a.e. The dominated convergence theorem therefore applies to the
integral in Eq.~\eqref{eq:auto:2}, and
\begin{equation} \label{eq:auto:5}
	\lim_{\afsnosup \to 0} D(\afs{})=D(0) \; .
\end{equation}
Applying the same estimates to the relativistic Hirshfeld weights, written in the form 
\begin{equation}
	w_A^\text{R} = \frac{ w_A^\NR \Lambda_A }{ \sum_B w_B^\NR\Lambda_B } 
\end{equation}
affords $w_A^\R(\,\bm\cdot\,;\afs{})\to w_A^\NR$ a.e.\ and
$\rho_A^\RAIM(\,\bm\cdot\,;\afs{})\to w_A^\NR\rho$ in $L^1$.

We can now establish the main result in this section, which is a relativistic extension of Theorem~\ref{thm:KL-stockholder}.
Let $s_A^{\mathrm{NR}},s_A^{\mathrm{R}}\in L^1(X)$ be normalized shape functions:
\begin{equation} \label{eq:shape_norm}
	s_A^{(\cdot)}\ge 0,\qquad \int_X s_A^{(\cdot)}(x)\,dx=1 \; .
\end{equation}
The specialization of Theorem~\ref{thm:KL-stockholder} to the Dirac limit, as stated in 
Lemma~\ref{lem:RAIM-KL} below, is the contribution of the present work, not the general uniqueness result.

\begin{lemma}[RAIM as the unique KL-stockholder partition in the Dirac limit]\label{lem:RAIM-KL}
Let $(X,\mathcal B,dx)=(\mathbb R^3,\mathrm{Borel},dx)$.
Let $\rho\in L^1(X)$ with $\rho\ge 0$ a.e.\ be the molecular electron density.
Fix an atom index set $\mathcal A=\{1,\dots,M\}$ and populations $N_A>0$.

For each $A\in\mathcal A$, let $s_A^\R(\,\bm\cdot\,;\afs{})\in L^1(X)$ satisfy
$s_A^\R\ge 0$ a.e. and $\int_X s_A^{\mathrm R}(x;\alpha_{\mathrm{fs}})\,dx=1$, and define the
(relativistic) promolecule component 
\begin{equation}
	\rho_{0,A}^\R(x;\afs{}) := N_A\,s_A^\R(x;\afs{})
\end{equation} 
and promolecule density 
\begin{equation}
	\rho_0^\R(x;\afs{}):=\sum_{B\in\mathcal A}\rho_{0,B}^\R(x;\afs{}) \; .
\end{equation}
Assume that 
\begin{equation}
	\rho_0^\R(x;\afs{})>0\quad\text{for a.e. }x\in\{\rho>0\} 
\end{equation}
and
\begin{equation}
	\int_X \rho\log\left(\frac{\rho}{\rho_0^{\mathrm R}}\right) dx < \infty \; .
\end{equation}
Define the feasible set
\begin{equation}
	\mathcal S:=\Big\{(\rho_A)_{A\in\mathcal A}\in \prod_{A\in\mathcal A}L^1(X):
	\rho_A\ge 0\ \text{a.e. and }\sum_{A\in\mathcal A}\rho_A=\rho\ \text{a.e.}\Big\} \; .
\end{equation}
and an extended-valued functional
\begin{equation}\label{eq:L-RAIM}
	\mathcal L_{\afsnosup}\big[(\rho_A)_{A\in\mathcal A}\big]
	:=\sum_{A\in\mathcal A}\int_X \rho_A(x)\log\left[\frac{\rho_A(x)}{\rho_{0,A}^{\mathrm R}(x;\afs{})}\right] dx 
\end{equation}
with conventions $0\log(0/c)=0$ for $c\ge 0$ and $t\log(t/0)=+\infty$ for $t>0$.
Then $\mathcal L_{\alpha_{\mathrm{fs}}}$ admits a unique minimizer over $\mathcal S$ (up to $dx$-null sets),
given by the RAIM stockholder partition:
\begin{equation}
	\rho_A^\RAIM(x;\afs{}) =\rho(x)\,w_A^\R(x;\afs{}) 
\end{equation}
with
\begin{equation}
	w_A^\R(x;\afs{}):=
	\frac{\rho_{0,A}^\R(x;\afs{})}{\rho_0^\R(x;\afs{})} \quad\text{for a.e. }x\in X \; .
\end{equation}
Moreover, the minimum value is
\begin{equation}
	D(\afs{})  :=
	\min_{(\rho_A)\in\mathcal S}\mathcal L_{\afsnosup}\big[(\rho_A)_A\big]
	=
		\int_X \rho(x)\log\left[\frac{\rho(x)}{\rho_0^\R(x;\afs{})}\right] dx 
\end{equation}
\end{lemma}

\begin{proof}
Apply Theorem~\ref{thm:KL-stockholder} with the choice
$q_A(x)=\rho_{0,A}^\R(x;\afs{})$ and $q(x)=\rho_0^\R(x;\afs{})$.
The hypothesis $q>0$ a.e. on $\{\rho>0\}$ is exactly the stated condition on $\rho_0^\R$.
The theorem yields the unique minimizer
$\rho_A^\star=\rho\,q_A/q=\rho\,\rho_{0,A}^\R/\rho_0^\R$, which is precisely
$\rho_A^\RAIM=\rho\,w_A^\R$.
Substituting the minimizer into $\mathcal L_{\afsnosup}$ gives
$\min \mathcal L_{\afsnosup}=\int \rho\log(\rho/\rho_0^\R)=D(\afs{})$.
\end{proof}

Thus, the RAIM weights $w_A^\R$ are the unique weights that minimize the total KL divergence from the promolecule components $\rho_{0,A}^\R$, subject to reproducing the fixed molecular density $\rho$.

\subsubsection{Perturbative Expansion}
\label{sec:math:KL:perturbative}

For the perturbative statement, assume in addition the small-$\afs{}$ expansion 
\begin{equation}\label{eq:auto:6}
	\Lambda_A(x;\afs{}) =1+\alpha_{\mathrm{fs}}^2 g_A(x)+R_A(x;\afs{})
\end{equation}
with
\begin{equation}
	\big|R_A(x;\afs{})\big| \le \afs{4} h_A(x) 
\end{equation}
holding uniformly for $|\afs{}| \le \alpha^{}_\text{fs,0}$, with 
measurable $g_A$ and $h_A$ satisfying
\begin{subequations}
\begin{align}
	\rho\,|g_A| &\in L^1(X) 
\\
	\rho\,|g_A|^2 &\in L^1(X) 
\\
	\rho\,h_A^2 &\in L^1(X) \; .
\end{align}
\end{subequations}
Since $\rho$ is a finite measure, $\rho h_A^2\in L^1$ implies $\rho h_A\in L^1$
by the Cauchy-Schwarz inequality, so the remainder is integrable against $\rho$ as well.

The even-power expansion in Eq.~\eqref{eq:auto:6}
is consistent with a leading relativistic correction
of order $c^{-2}$ at fixed electronic configuration. The uniform remainder and
integrability conditions are additional assumptions on the reference family;
a finite nuclear model and a finite orbital basis alone do not establish them.
The normalization
$\int_X s_A^{\mathrm{NR}}\Lambda_A\,dx=1$, combined with Eq.~\eqref{eq:auto:6} and the
integrability condition $s_A^{\mathrm{NR}}h_A\in L^1(X)$, forces the mean-zero condition 
\begin{equation}
	\int_X s_A^{\mathrm{NR}}(x) \; g_A(x)\,dx=0 \; .
\end{equation}
The leading correction is a genuine reshaping of the reference atom, not a rescaling.  Defining
\begin{equation}\label{eq:bar-g}
	\bar g(x):=\sum_{A\in\mathcal{A}} w_A^\NR(x)\; g_A(x) \; ,
\end{equation}
the sign of $\int\rho\,\bar g\,dx$ 
is determined by where the molecular density overlaps the expansion ($g_A>0$) versus depletion ($g_A<0$) regions of each atom.

\begin{theorem}[Perturbative expansion of $D(\afs{})$]\label{thm:expansion}
Assume $D(0)<\infty$, the uniform bounds in Eq.~\eqref{eq:auto:4}, and the expansion in Eq.~\eqref{eq:auto:6}
with the stated integrability conditions. Then, with $\bar g$ as defined in Eq.~\eqref{eq:bar-g}, one has
\begin{equation}\label{eq:auto:9} 
	D(\alpha_{\mathrm{fs}}) =
	D(0) - \afs{2} \int_X \rho(x)\,\bar g(x)\,dx
	+ O(\afs{4})
	\qquad \text{as}~\alpha_{\mathrm{fs}}\to0\; .
\end{equation}
\end{theorem}

\begin{proof}
Substitute the expansion in Eq.~\eqref{eq:auto:6} into the convex combination and write
\begin{equation}
	\sum_A w_A^\NR \Lambda_A = 1+u
\end{equation}
where
\begin{equation} 
	u(x;\afs{}) := \afs{2}\; \bar g(x)+R(x;\afs{})  \; .
\end{equation}
Defining $h := \sum_A h_A$, the aggregate remainder obeys
\begin{equation}
	|R|\le\afs{4} \sum_A w_A^\NR h_A\le \afs{4} h \; , 
\end{equation}
with $\rho h\in L^1(X)$.

It is not enough to observe that $u$ is formally of order $\afs{2}$.  The 
coefficient $\bar g$ is only integrable against $\rho$, not bounded, so $u$
need not be pointwise small and the naive expansion of the logarithm is not
justified. What the uniform bounds in Eq.~\eqref{eq:auto:4} supply instead is a
two-sided bound independent of $x$:
\begin{equation}
	u(x;\afs{}) = \sum_A w_A^\NR \Lambda_A - 1 \in [m-1,\,M-1]
	\quad\text{a.e.},\qquad |\afs{} |\le \alpha^{}_\text{fs,max} \; ,
\end{equation}
where $[m-1,M-1]$ is a compact subinterval of $(-1,\infty)$.
On that interval, the map $\phi(u):=\log(1+u)$ is smooth, with
$|\phi''(u)|=(1+u)^{-2}\le m^{-2}$. Taylor's theorem with Lagrange remainder
therefore gives the pointwise estimate
\begin{equation}
	\big|\log(1+u)-u\big|\le \frac{u^2}{2m^2}\qquad\text{a.e.} \; .
\end{equation}

It remains to check that $\rho u^2$ is integrable. By Jensen's inequality applied
to the convex combination, $\bar g^2\le\sum_A w_A^{\mathrm{NR}}g_A^2$ so
$\rho\,\bar g^2\in L^1(X)$. Combining this with
$(\afs{2}\bar g+R)^2\le 2\afs{4}\bar g^2+2R^2$ and $|R|\le\afs{4} h$, one obtains
\begin{equation}
	\rho\,u^2 \le 2\, \afs{4}\, \rho\,\bar g^2 + 2\,\afs{8}\,\rho\,h^2 \; .
\end{equation}
Both terms on the right are integrable by hypothesis.
Integrating the pointwise Taylor estimate against $\rho$ and inserting this bound affords 
\begin{align}
\begin{aligned}
	\int_X\rho \bigg|\log\Big(\sum_A w_A^\NR \Lambda_A\Big)-u\bigg| dx
	&\le\frac{1}{2m^2}\int_X\rho\,u^2\,dx
\\ &
	\le\frac{\afs{4}}{m^2}\int_X\rho\,\bar g^2\,dx+O(\afs{8})
	=O(\afs{4}) \; .
\end{aligned}
\end{align}
Finally, the first-order term is integrable against $\rho$ because
$\rho|g_A|\in L^1(X)$, and the remainder integral is bounded by
$\afs{4}\int_X\rho h\,dx = O(\afs{4})$. Substituting into Eq.~\eqref{eq:auto:2} affords Eq.~\eqref{eq:auto:9}.
\end{proof}

The coefficient of $\afs{2}$ in Eq.~\eqref{eq:auto:6} need not be nonzero. In particular, if
$\rho=\rho^{0,\NR}$, then normalization gives
\begin{equation}
	\int_X\rho\,\bar g\,dx=\sum_A N_A\int_X s_A^{\NR}g_A\,dx=0
\end{equation} 
and $D(\afs{})-D(0)=O(\afs{4})$.

The deformation $\Lambda_A(x;\afs{})$ acts solely on the reference measure used to construct 
partition-dependent atomic response parameters.  Dispersion corrections such as MBD, TS-vdW, or XDM 
are applied \textit{a~posteriori}, after the self-consistent field (SCF) equations have been solved for $\rho$.
Thus, the SCF equations (hence $\rho$) are unchanged by the deformation although the total energy 
$E = E_\text{SCF} + E_\text{disp}$ (including the dispersion correction) generally does change.

\section{Computational Results}\label{sec:calcs}
We now quantify the relativistic deformations, $\Lambda_A$ [Eq.~\eqref{eq:LambdaA}].  According to Lemma~\ref{lem:RAIM-KL}, 
the RAIM partition is the unique KL-stockholder minimizer of the deformation functional in Eq.~\eqref{eq:L-RAIM}.
The variational machinery is unchanged from the non-relativistic case; RAIM corrects the reference measure, not the response theory.

\subsection{Procedure}

Scalar contraction factors
\begin{equation}\label{eq:Lambdabar}
	\bar{\Lambda}_A := \left(
		\frac{ \langle r^3\rangle_\A^\R }{ \langle r^3\rangle_\A^\NR }
	\right)^{1/3}
\end{equation}
were computed using the \dirac\ program~\cite{DIRAC}, for all 118 elements plus some ionic charge states. These calculations employ
Dyall all-electron triple-$\zeta$ (\texttt{ae3z}) basis sets~\cite{DyaFae07} with a finite Gaussian nuclear distribution~\cite{VisDya97}.

Isolated atoms and ions were placed at the origin within a four-component relativistic
framework~\cite{SauJen03,Sau11}. Two calculations were performed per species: a
four-component DHF calculation (for $\rho^\R$) and a non-relativistic HF calculation (to obtain $\rho^\NR$),
both with the same program, basis set, and nuclear model~\cite{DIRAC}.
Where the occupation could be specified in the usual way (see below), $\rho^\NR$ was obtained from the
non-relativistic Hamiltonian as implemented in \dirac\ (keyword \texttt{.NONREL}).
For the remaining open-shell species, whose occupations were specified in terms of
relativistic quantum numbers, the four-component DHF calculation was instead repeated with
the speed of light raised from $c \approx 137$~a.u.\ to $c=10^4$~a.u.\ (keyword \texttt{.CVALUE}).
The Dirac-Coulomb Hamiltonian reduces to the Schr\"odinger Hamiltonian as $c\to\infty$,
with leading corrections that are $O(c^{-2})$~\cite{DyaFae07}, so the larger value of $c$
approximates the non-relativistic limit.  It suppresses the leading correction
by a factor of $(137/10^4)^2 \approx 2\times 10^{-4}$ relative the physical value of $c$.

This factor is not an error bound on the cubic radial moment $\langle r^3 \rangle$, however.  
The cubic radial moment $\langle r^3 \rangle$ weights the outer density through the
radial integrand $r^5\rho_A(r)$, so its convergence must be assessed separately from particle-number conservation. 
A finite Gaussian nuclear charge distribution (\texttt{.NUCMOD 2}) was used for both density calculations, so that
point-nucleus singularities are absent and the ratio is controlled near the
origin. Two-electron integrals were screened at $10^{-12}$~a.u.\ and no frozen-core approximations were used.

Three SCF protocols were required. Closed shells use a standard closed-shell SCF. 
Open shells with unambiguous ground-state occupation use
an average-of-configurations (AOC) approach with fractional occupation $m/g$ over the shell
degeneracy $g$. For Au($5\text{d}^{10}6\text{s}^1$), the closed-shell population is 78 and
the remaining electron is distributed over the two 6s spinors with
occupation 1/2 each; the 5d shell remains fully occupied. An open
$4\text{f}^7$ shell has occupation 7/14 per spinor.
Finally, for $Z\ge47$ a level shift of $0.3$~a.u.\ and Fock damping of $0.15$
were needed to prevent oscillations between competing configurations.  For La, near-degeneracy of the 
$5\text{d}^1$ and $4\text{f}^1$ configurations requires 
more aggressive intervention ($0.8$~a.u.\ level shift, damping coefficient of $0.40$, and 800 iterations).
The remaining 45 elements with ambiguous ground states (including all of the lanthanides and actinides along with 
selected 4d metals, Cs, Ru$^+$, and all transactinides) were converged in full atomic symmetry with
the occupation specified separately for each value of the relativistic
angular quantum number $\kappa$, where $\kappa=-(\ell+1)$ for $j=\ell+\tfrac{1}{2}$ and
$\kappa=\ell$ for $j=\ell-\tfrac{1}{2}$~\cite{DyaFae07}. [Each $\kappa$ labels one
spin--orbit component of a subshell (s$_{1/2}$, p$_{1/2}$, p$_{3/2}$, d$_{3/2}$, \ldots).]
This is the $\kappa$-selection option of \dirac\ (keyword \texttt{.KPSELE})~\cite{DIRAC}. 
By default, the SCF fills spinors in order of energy, which allows electrons to move
between near-degenerate d and f subshells from one iteration to the next; fixing
the number of electrons in each $\kappa$ channel prevents that redistribution.
The same $\kappa$ occupations were used for both the DHF and the non-relativistic HF calculations, in order 
to compare the two Hamiltonians at fixed channel occupations.

Radial densities were extracted with \texttt{**VISUAL}/\texttt{.DENSITY DFCOEF}.
Particle-number conservation
was verified to within $10^{-4}$ for
every species, and $\langle r^3\rangle$ was evaluated by trapezoidal quadrature on
the same grid for both the relativistic and non-relativistic Hamiltonians. 
Using the same basis and grid can reduce
correlated numerical errors in the ratio $\bar{\Lambda}_A$; the particle-number check alone does
not establish convergence of the cubic radial moment.

\begin{figure}[t]
\centering
\scalebox{1.0}{\includegraphics{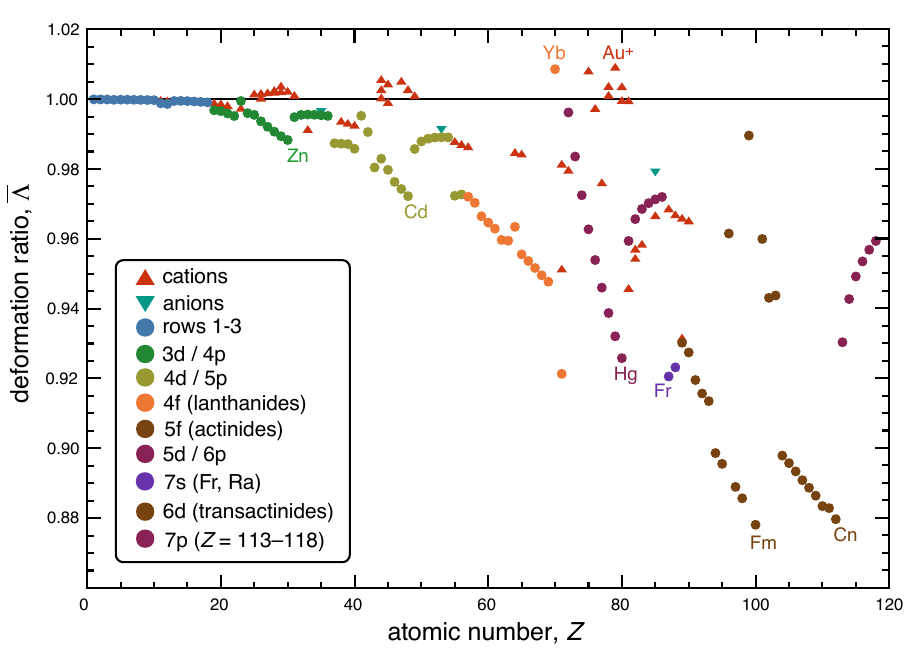}}
\caption{
	Deformation ratios $\bar{\Lambda}$ [Eq.~\protect\eqref{eq:Lambdabar}] versus atomic number, 
	for 118 elements and 64 ions.  DHF and non-relativistic HF 
	calculations using all-electron, triple-$\zeta$ Dyall basis sets were performed using the \dirac\ program.
}\label{fig:lambda}
\end{figure}

\subsection{Results}

Numerical results in Fig.~\ref{fig:lambda} illustrate the atomic ratios $\bar{\Lambda}$ for all neutral atoms up to $Z=118$ along with some ions.
These data reveal a rich periodic structure with chemically interpretable origin.
Gold contracts because the direct relativistic stabilization of the 6s orbital dominates, pulling electron density inward: 
$\bar{\Lambda}_\text{Au} = 0.932$, $\bar{\Lambda}_\text{Hg} = 0.926$, $\bar{\Lambda}_\text{Fr} = 0.921$, and 
$\bar{\Lambda}_\text{Fm} = 0.878$.   The Ru
cations do the opposite. Ionization of Ru ($4\text{d}^7 5\text{s}^1$, $\bar{\Lambda} = 0.983$) removes the contracted 5s electron
and $\bar{\Lambda}$ crosses unity, $\bar{\Lambda}_{\text{Ru}^+} = 1.005$; further ionization removes 4d electrons and $\bar{\Lambda}$
relaxes back toward unity, with $\bar{\Lambda}_{\text{Ru}^{2+}} = 1.003$ and $\bar{\Lambda}_{\text{Ru}^{3+}} = 1.000$. 
The same inversion holds for all coinage metal cations: 
$\bar{\Lambda}_{\text{Cu}^+} = 1.004$, 
$\bar{\Lambda}_{\text{Ag}^+} = 1.005$, and 
$\bar{\Lambda}_{\text{Au}^+} = 1.009$; see Table~\ref{tab:ions}.
Ionization removes the contracted $n$s electron, leaving a d$^{10}$ core whose indirect relativistic expansion is no longer offset~\cite{Pyy12}.
The trivalent lanthanide cations contract modestly but systematically, with the following values for $\bar{\Lambda}$: Gd$^{3+}$
($4\text{f}^7$, $0.9846$), Tb$^{3+}$ ($4\text{f}^8$, $0.9841$), Lu$^{3+}$ ($4\text{f}^{14}$,
$0.9812$). Stripping the $6\text{s}^2$ electron pair removes the direct contraction channel that
dominates the neutral atoms; what survives is the residual contraction of the
$5\text{s}^25\text{p}^6$ core, deepening along the series as the nuclear charge grows and the poorly screened 4f shell fills.
The lesson is consequential: neutral-atom $\bar\Lambda$ values are inapplicable to charged metal centers, and only a relativistic
DHF calculation on the ion affords a correct reference for such species.

\begin{table}[t] 
\centering
\caption{
	Selected deformation ratios for ions.$^a$
}\label{tab:ions}
\vspace{0.25cm}
\renewcommand{\arraystretch}{1.1}
\begin{threeparttable}
\begin{tabular}{l @{\hspace{0.35cm}} c @{\hspace{0.35cm}} c @{\hspace{2cm}} lcc}
\hline\hline
ion & configuration & $\bar{\Lambda}$$^b$ &
ion & configuration & $\bar{\Lambda}$$^b$ \\
\hline
Cu$^+$    	& $[\text{Ar}]3\text{d}^{10}$        			& 1.0036 & Tb$^{3+}$ 	& $4\text{f}^8$                        			& 0.9841 \\
Ag$^+$    	& $[\text{Kr}]4\text{d}^{10}$        			& 1.0050 & Gd$^{3+}$ 	& $4\text{f}^7$                        			& 0.9846 \\
Au$^+$    	& $[\text{Xe}]4\text{f}^{14}5\text{d}^{10}$ 	& 1.0090 & Lu$^{3+}$ 	& $4\text{f}^{14}$                   			& 0.9812 \\
Re$^+$    		& $5\text{d}^5$                        				& 1.0079 & Tl$^+$    	& $6\text{s}^2$                        			& 0.9456 \\
Pt$^{2+}$ 	& $5\text{d}^8$                        				& 1.0009 & Pb$^{2+}$ 	& $6\text{s}^2$                        			& 0.9542 \\
Ru$^{3+}$ 	& $4\text{d}^5$                       				& 1.0002 & Ac$^{3+}$ 	& $[\text{Rn}]$         	      			& 0.9657 \\
Rh$^{3+}$ 	& $4\text{d}^6$                        				& 0.9989 & I$^-$     		& $5\text{p}^6$                        			& 0.9916 \\
Os$^{3+}$ 	& $5\text{d}^5$                        				& 0.9972 & Br$^-$    	& $4\text{p}^6$                        			& 0.9967 \\
Ir$^{3+}$ 	& $5\text{d}^6$                        				& 0.9759 & F$^-$     	& $2\text{p}^6$                        			& 1.0000 \\
\hline\hline
\end{tabular}
\begin{tablenotes}[flushleft]\footnotesize
	\item
	$^a$Based on DHF and non-relativistic HF calculations in all-electron Dyall triple-$\zeta$ basis sets, 
	performed using the \dirac\ program.
	$^b$Eq.~\protect\eqref{eq:Lambdabar}.
\end{tablenotes}
\end{threeparttable}
\end{table}

The periodic structure extends coherently beyond the 5d maximum. The 6p block
recovers as p-filling dilutes the contracted s/d core
($\bar\Lambda_{\mathrm{Tl}}=0.959\to\bar\Lambda_{\mathrm{Rn}}=0.972$), and Fm
($5\text{f}^{12}$, $\bar\Lambda=0.878$) is the largest contraction in either f-block.
The 4f series contracts with shell filling from Gd ($4\text{f}^7$, $\bar{\Lambda}=0.9635$) to Lu
($\bar{\Lambda}=0.9515$), interrupted by a slight expansion at closed-shell Yb ($4\text{f}^{14}$, $\bar{\Lambda}=1.0086$). 
For $Z\le18$ all values satisfy $\bar\Lambda>0.999$, and the 3d row
reaches only $\bar\Lambda\approx0.99$. From the perspective of a relativistic
dispersion model, all of organic chemistry is effectively in the $\bar\Lambda=1$ limit.
Neutral Tb and Gd admit multiple low-lying SCF solutions giving
qualitatively different values ($1.100$ and $0.963$ at default AOC);
the trivalent cations,
whose relativistic and non-relativistic occupations coincide by construction,
are the better-defined reference.

The chemical consequence enters quantitatively through the generalized TS scaling factor
\begin{equation}\label{eq:gammaA}
	\gamma^{}_A = \frac{N_A}{N_A^0}\,\frac{\langle r^3\rangle_\A}{\langle r^3\rangle_\A^0} 
\end{equation}
[\textit{cf}.\ Eq.~\eqref{eq:TS}].  The scaling factor $\gamma^{}_A$ 
controls the atoms-in-molecule partition of the static polarizability:
$\alpha^{}_A = \gamma^{}_A \alpha_A^\text{ref}$, $C_{6,AA} = \gamma_A^2 C_{6,AA}^\text{ref}$, and 
$R_{\text{vdw},A} = \gamma_A^{1/3} R_{\text{vdW},A}^\text{ref}$.
This response scaling uses a scalar volume descriptor of the partition.
The contraction factor $\bar\Lambda_A$ summarizes the change in the free atom's 
cubic moment; specifying how it modifies response parameters is an additional
modeling choice beyond the KL characterization of the density partition.

\begin{table}[t]
\centering
\caption{
	Four-component DHF static dipole polarizabilities
	($\alpha$) and homonuclear dispersion coefficients ($C_6$), in atomic units.$^a$
}\label{tab:c6}
\vspace{0.25cm}
\renewcommand{\arraystretch}{1.1}
\begin{threeparttable}
\begin{tabular}{l @{\hspace{0.35cm}} d d c| @{\hspace{0.75cm}} ldd}
\hline\hline
species & \mc{1}{r}{$\alpha$~~~} & \mc{1}{r}{$C_6$~~~} && species & \mc{1}{r}{$\alpha$~~~} & \mc{1}{r}{$C_6$~~~} \\
\hline
He        	&  1.3070 	&   1.367 		&& Cu$^+$    		&  5.0119 	&  15.825 \\
Ne        	&  2.1337 	&   4.991 		&& Ag$^+$    		&  8.5878 	&  50.473 \\
Ar        	& 10.4332 	&  60.802 	&& Mg$^{2+}$ 	&  0.4696 	&   0.558 \\
Kr        	& 16.2546 	& 124.824 	&& Ca$^{2+}$ 	&  3.1876 	&  10.599 \\
Zn        	& 50.7082 	& 399.464 	&& Sr$^{2+}$ 		&  5.8069 	&  28.402 \\
Cd        	& 63.5135 	& 610.475 	&& Ba$^{2+}$ 	& 10.5904 	&  75.712 \\
Li$^+$    	&  0.1893 	&   0.077 		&& Zn$^{2+}$ 		&  2.2373 	&   5.851 \\
Na$^+$    &  0.9315 	&   1.469 		&& Cd$^{2+}$ 	&  4.9240 	&  24.734 \\
K$^+$     	&  5.3924 	&  23.071 	&& Yb$^{2+}$ 		&  6.3598 	&  31.870 \\
Rb$^+$   	&  9.0661 	&  54.178 	&& Lu$^{3+}$ 		&  4.1682 	&  18.010 \\
Cs$^+$   	& 15.7837 	& 133.536 	&&           &         	&         \\
\hline\hline
\end{tabular}
\begin{tablenotes}[flushleft]\footnotesize
	\item
	$^a$Four-component DHF linear response to an electric
	dipole perturbation, uncontracted Dyall \texttt{av3z} basis sets, finite
	nuclear model, evaluated at imaginary frequencies on a transformed
	Gauss-Legendre grid and integrated by Eq.~\protect\eqref{eq:CP}. The isotropic
	average $\alpha(\cmplxi\omega)=\mathrm{tr}(\bm\alpha)/3$ is used throughout. The zero-frequency 
	value is reported as $\alpha$ and $\omega=0$ is not a quadrature node.
	Values are uncorrelated; see the text.
\end{tablenotes}
\end{threeparttable}
\end{table}

The reference data on which that handle acts can be placed on the same
four-component footing. Table~\ref{tab:c6} collects static dipole
polarizabilities and homonuclear $C_6$ coefficients obtained from
four-component DHF linear response at imaginary frequency, using the 
Casimir-Polder integral in the uncoupled approximation~\cite{CasPol48,Lon65,Her21a},
\begin{equation}\label{eq:CP}
	C_{6,AB} = \frac{3\hbar}{\pi}\int_0^\infty
	\alpha_A(\cmplxi\omega)\,\alpha_B(\cmplxi\omega)\,d\omega \; .
\end{equation}
Numerical data in Table~\ref{tab:c6} represent
closed-shell atoms and ions whose response calculations converged.
Augmented basis sets were used for the response (uncontracted Dyall \texttt{av3z}), 
in order to capture diffuse excitations.
These are DHF response values without any post-Hartree-Fock treatment of electron correlation, which 
might change them substantially in either direction. They are not
offered as a replacement for correlated reference data. Their purpose is
internal consistency: a $C_6$ computed
on the same Hamiltonian and nuclear model as the densities from which the
partition is built. Fifteen of
the twenty-one species are cations, for which correlated reference data are
sparse; this is the regime in which the present construction is most needed.

\section{Conclusions}
We have developed a relativistic version of Hirshfeld's atoms-in-molecules construction and placed it on a formal 
footing as a KL stockholder problem with a Dirac-consistent reference measure. As in the non-relativistic case, the 
KL minimizer of the relativistic functional is unique (Lemma~\ref{lem:RAIM-KL}).  For the relativistic case, the 
minimum divergence admits an asymptotic
expansion to second order in the fine-structure constant under explicit
regularity (Theorem~\ref{thm:expansion}).  Finally, the optimized quantum Drude
oscillator constraint that closes the dispersion model from
$\{\alpha,C_6,\Rvdw{}\}$ admits a closed-form Lambert-$W$ solution
(Theorem~\ref{thm:LambertW}) that identifies the compact oscillator with the $W_{-1}$ branch.

This RAIM construction modifies the prior, not the variational principle;
information-theoretic uniqueness, $\sigma$-finiteness, and the dominated
convergence machinery transfer unchanged. What enters is a one-parameter
deformation of the reference measure, the small parameter being the
fine-structure constant. The Lambert-$W$ result is independent of the KL
extension and supplies an analytic width for a single isotropic oscillator
subject to the stated optimized QDO constraint.
Going forward, we will use this formalism to build a dispersion model; the damping
parameterization and its validation on the MBD model as used in symmetry-adapted perturbation 
theory~\cite{CarLaoLiu19,LiuCarHer19,CarLaoHer21,HerGraLiu25} is in progress.

\vspace{0.5cm}

\noindent\textbf{Acknowledgements}
K.P. thanks the late Prof.\ John F. Stanton, whose early teaching in coupled cluster theory 
and the algebraic reduction of many-body operators shapes this work more than any citation could record.
Calculations were performed at the Ohio Supercomputer Center~\cite{OSC}. \\


\noindent\textbf{Author contributions}
K.P. performed the conceptualization and investigation (including methodology, software, and validation), 
data curation, and analysis, and wrote the original draft. J.M.H contributed to the conceptualization, led the
project administration, supervision, and funding acquisition, and edited the manuscript. 
Both authors approved the final manuscript. \\

\noindent\textbf{Funding}
This work was supported by the U.S. Department of Energy, Office of Basic Energy Sciences,
Division of Chemical Sciences, Geosciences, and Biosciences under Award No.\ DE-SC0008550.  \\


\section*{Declarations}

\noindent \textbf{Conflict of interest} 
The authors declare no conflict of interest.


\def\bibsection{\section*{\refname}} 


\end{document}